\documentclass{elsarticle}
\journal{}

\usepackage[utf8]{inputenc}
\usepackage{color}
\usepackage{graphicx}
\usepackage{amsmath}
\usepackage{amssymb}

\newtheorem{theorem}{Theorem}[section]

\newtheorem{lemma}{Lemma}[section]

\newenvironment{proof}{\par\noindent{\it Proof:\/}}{}

\newcommand{\Access}{\mbox{$\mathit{access}$}}
\newcommand{\Rank}{\mbox{$\mathit{rank}$}}
\newcommand{\Select}{\mbox{$\mathit{select}$}}

\newcommand{\Index}{\mbox{$\mathit{index}$}}
\newcommand{\Active}{\mbox{$\mathit{active?}$}}

\newcommand{\floors}[1]{\left\lfloor #1 \right\rfloor}
\newcommand{\ceils}[1]{\left\lceil #1 \right\rceil}

\newcommand{\twodots}{\,.\,.\,}
\newcommand{\sequence}[1]{\left\langle#1\right\rangle}

\newcommand{\ylatyhjaa}{\mbox{${}^{\mbox{\rule{0cm}{1.25ex}}}$}}

\begin{document}
\begin{frontmatter}
\title{Selection from read-only memory \\ with limited workspace \tnoteref{t1}}

\tnotetext[t1]{The results of this paper were presented in preliminary
  form at the 19th Annual International Computing and Combinatorics
  Conference held in Hangzhou, China, in June 2013.}

\author[AU]{Amr Elmasry}
\ead{elmasry@alexu.edu.eg}
\author[KU]{Daniel Dahl Juhl}
\ead{juhl.daniel@gmail.com}
\author[KU]{Jyrki Katajainen}
\ead{jyrki@di.ku.dk}
\author[SNU]{Srinivasa Rao Satti\fnref{ssr-thanks}}
\ead{ssrao@cse.snu.ac.kr}

\fntext[ssr-thanks]{The research of this author was partly supported
  by Basic Science Research Program through the National Research
  Foundation of Korea (NRF) funded by the Ministry of Education,
  Science and Technology (grant number 2012-0008241).}

\address[AU]{Department of Computer Engineering and Systems, Alexandria University,
Egypt}
\address[KU]{Department of Computer Science, University of Copenhagen, Denmark}
\address[SNU]{School of Computer Science and Engineering, Seoul National University, South Korea}

\begin{abstract}
Given an unordered array of $N$ elements drawn from a totally ordered
set and an integer $k$ in the range from $1$ to $N$, in the classic
selection problem the task is to find the $k$-th smallest element in
the array.  We study the complexity of this problem in the
space-restricted random-access model: The input array is
stored on read-only memory, and the algorithm has access to a 
limited amount of workspace. We prove that the linear-time prune-and-search
algorithm---presented in most textbooks on algorithms---can be
modified to use $\Theta(N)$ bits instead of $\Theta(N)$ words of extra
space. Prior to our work, the best known algorithm by Frederickson
could perform the task with $\Theta(N)$ bits of extra space in $O(N
\lg^{*} N)$ time. Our result separates the space-restricted
random-access model and the multi-pass streaming model, since we can
surpass the $\Omega(N \lg^{*} N)$ lower bound known for the latter
model.  We also generalize our algorithm for the case when the size of
the workspace is $\Theta(S)$ bits, where $\lg^3{N} \leq S \leq N$.  The
running time of our generalized algorithm is $O(N \lg^{*}(N/S) + N
(\lg N) / \lg{} S)$, slightly improving over the $O(N \lg^{*}(N
(\lg N)/S) + N (\lg N) / \lg{} S)$ bound of Frederickson's algorithm.  
To obtain the improvements mentioned above, we developed a new data structure,
called the wavelet stack, that we use for repeated pruning. 
We expect the wavelet stack to be a useful tool in other applications as well.
\end{abstract}
\end{frontmatter}


\section{Introduction}

Let $\sequence{x_1, x_2, \ldots, x_N}$ be a sequence of unordered elements, from a totally ordered set, stored in an array.
Given an integer $k$ in the range $1 \leq k \leq N$, 
in the selection problem we want to find the $k$-th smallest of these
elements.  Without loss of generality, we can assume that the
elements are distinct (since in the case of equal elements, the
indices can be used to distinguish the elements).  That is, the output
will be a single index $j$ with a guarantee that $k - 1$ elements are
smaller than $x_j$ and $N - k$ elements are larger than $x_j$.

The asymptotic time complexity of the selection problem was settled to be
$\Theta(N)$ by Blum et al.~\cite{BFPRT73} in their celebrated article
from 1973. Here we study the problem in the \emph{space-restricted
  random-access model}, where the input elements are given in a
read-only array and a limited amount of additional workspace is
available.  We first focus on the case where the amount of
workspace is $\Theta(N)$ bits, and subsequently consider the general
case with a more limited workspace. Surprisingly, although the
time-space trade-offs for selection in this read-only space-restricted setting
have been studied in several papers
\cite{Cha10,Fre87,MP80,MR96,RR99}, its exact complexity is still not fully
resolved (even after our study).

We start by describing an algorithm that
solves the selection problem using $O(\lg^2 N)$ extra bits (or a logarithmic
number of machine words).  Although this algorithm does not match the 
time complexity of the best known algorithm using this amount of workspace, it
illustrates some of the difficulties involved in designing algorithms for this
model of computation, and also describes a few techniques that we shall use
in our main algorithm. 
As in many other selection algorithms, we maintain two
indices that specify the so-called \emph{filters}. The elements whose values 
fall within the range of the two filters are still possible candidates for being the $k$-th smallest
element. We say that the elements within the range of the filters are \emph{active}. 
At the beginning of the algorithm, we scan the input and initialize the two filters to 
be the minimum and maximum elements. If $k=1$ or $N$, we are done; otherwise, all elements are active
except the filters.

The algorithm proceeds recursively. Each recursive call takes the following 
parameters: (i) a contiguous segment of $M$ elements, specified by its
first and last array indices, (ii) a pair of filters, specified by their positions in the 
input array, and (iii) a parameter $k$; and returns the $k$-th smallest 
element among the active elements within the segment. 
The recursive call proceeds by dividing the segment into $t$ {\em zones} of contiguous elements:
each zone except the last constitutes $\ceils{M/t}$ elements, and the remaining elements 
form the last zone. We then check which of the $t$ zones contains the
largest number of active elements (this idea is from \cite{MR96}); we say that this
zone is \emph{heavy}. Next, we select the median of the heavy zone
recursively. Let $x_{i1}$ and $x_{i2}$ be the filters (before the
recursive calls) and $x_m$ the
median found. After this first recursive call, we scan through the elements
in the current segment to determine whether the $k$-th smallest element
is in the interval $(x_{i1}\twodots x_{m})$, is equal to $x_m$, or is in
the interval $(x_m\twodots x_{i2})$.  If $x_m$ is the $k$-th smallest element, we
return $x_m$ as output.  In the other two cases, we update the 
filters. In the last case, we set $k$ to $k -\alpha$ if $\alpha$ smaller active elements were eliminated. 
Since $x_m$ is the median of the active elements in the heavy zone,
at least $\frac{1}{2t}$-th of the active elements will be removed from
further consideration.  Finally, we perform the second recursive call to find the $k$-th smallest
of the remaining active elements in the whole segment.

For each subproblem we store: the segment boundaries, the present
filters, and a bit indicating which of the recursive calls---the first
or the second---is invoked. When $t$ is a constant, the maximum depth
of the recursion stack is $O(\lg N)$.  Thus, the overall workspace
used is $O(\lg^2 N)$ bits.

We use $n$ to denote the number of active elements, and $M$ to denote the size of the contiguous segment 
where the elements reside. The worst-case running time can be
described using the recurrence:
$$T(n, M) \leq \left\{
\begin{array}{l}
c_1 \cdot M \mbox{~~if $n < n_0$} \\
T\left(n', \ceils{M/t} \right)
+ 
T\left(n'', M\right)
+
c_2 \cdot M  \mbox{~~if $n \ge n_0\,,$}
\end{array}
\right.$$
\noindent
where $c_1, c_2$ and $n_0$ are positive constants, and $n'$ and $n''$
denote the number of active elements of the subproblems in the first and second recursive
calls, respectively. Since $n/t \leq n' \leq n$, let $n'=r \cdot n$ where $1/t \leq r \leq 1$. 
It follows that $n'' \leq (1-r/2) \cdot n$.  
When we set $t=16$, for the case $n \geq n_0$ where $n_0$ is big enough, one can show by substitution that $T(n,M) \leq c \cdot (\sqrt{n}-1) \cdot M$ for a positive constant $c$ big enough compared to $c_1$ and $c_2$. Indeed, from the induction hypothesis, $T\left(n', \ceils{M/16} \right) \leq c \cdot (\sqrt{r\cdot n} - 1) \cdot \ceils{M/16}$
and $T\left(n'', M\right) \leq c \cdot (\sqrt{(1-r/2)\cdot n} - 1) \cdot M$. 
Since $\sqrt{1-r/2} < 1-r/4$, thus $\sqrt{r}/16 + \sqrt{1-r/2} < 1$
when $1/16 \leq r \leq 1$, and the induction hypothesis is validated.
For these settings, the running time of the algorithm is therefore $O(N^{3/2})$.

The above algorithm highlights several aspects that are important
for algorithms designed for the space-bounded random-access
machine. Since we cannot move or modify elements, 
neither can we utilize enough space to store the indices of the active elements,
we have to scan the read-only array several times and pass over already eliminated elements. 
Due to the limited memory resources, 
we cannot store everything, and sometimes have to recompute
information that has already been computed previously. Also because of the limited 
workspace, it might be necessary to resort to some bit tricks to save space and time.

The performance of the best known selection algorithms is summarized in Table
\ref{table:known-algorithms} for workspace sizes specified as a function of the number of elements. 
In this paper we improve the known
results when the amount of extra space is $\Theta(N)$ bits, by giving a
new implementation for an adapted version of the algorithm of Blum et al.~\cite{BFPRT73} that also runs in $\Theta(N)$ time.
For the general case of $\Theta(S)$ bits of workspace, the best known algorithm is
that of Frederickson~\cite{Fre87}. The running time of Frederickson's algorithm 
is $O(N \lg^*(N (\lg N)/S) + N (\lg N) / \lg S)$ when $S = \Omega(\lg^3 N)$. 
We generalize our algorithm to run in $O(N \lg^*(N/S) + N (\lg N) / \lg S)$ time
and use $\Theta(S)$ bits, for any $\lg^3 N \le S \le N$, 
and thus improve Frederickson's algorithm (for a narrow range of $S$).

\begin{table}[t!]
\caption{\label{table:known-algorithms}The best known algorithms for
  selecting the $k$-th smallest of $N$ elements in the space-restricted
  random-access model; $N$ is the number of elements of the read-only input and $k$
  is an arbitrary integer between 1 and $N$. 
  The algorithms (except that in \cite{BFPRT73}) work for a larger range of workspace, 
  but we give their running times only
  for these specific values.}
\begin{center}
\begin{small}
\begin{tabular}{|l|c|c|c}
\cline{1-3}
\multicolumn{1}{|c|}{\textbf{Inventors}\ylatyhjaa} & \textbf{Workspace in bits} &
\textbf{Running time}\\ 
\cline{1-3}
Munro and Raman\ylatyhjaa{} \cite{MR96} & $\Theta(\lg N)\strut$ & 
$O(N^{1 + \varepsilon})$ & \\
Raman and Ramnath \cite{RR99} & $\Theta(\lg^2 N)$ & 
$O(N\lg^2 N)$ & \\
Frederickson \cite{Fre87} & $\Theta(\lg^3 N)$ &
$O(N \lg N/\lg\lg N)$ & \\
Frederickson \cite{Fre87} & $\Theta(N)$ &
$O(N \lg^{*} N)$ & \\
Blum et al.~\cite{BFPRT73} & $\Theta(N\lg N)$ &
$\Theta(N)$ & \\
\cline{1-3}
Elmasry et al.~[this paper] & $\Theta(N)$ & $\Theta(N)$ & \\
\cline{1-3}
\end{tabular}
\end{small}
\end{center}
\vspace{-.2in}
\end{table}

In the literature two main models of
computation have been considered when handling read-only data: the multi-pass streaming model
\cite{Cha10,Fre87,MP80} and the space-restricted
random-access model (that is used in this paper) \cite{Fre87,MR96,RR99}.  
The essential difference is that,
in the streaming model, the read-only input must be accessed
sequentially but multiple scans of the entire input are allowed; in
addition to the running time, the number of scans performed would be
an optimization target.  Chan~\cite{Cha10} proved that
Frederickson's algorithm is asymptotically optimal for the selection
problem in the multi-pass streaming model.  He questioned whether this
lower bound would also hold in the space-restricted random-access
model.  We answer this question negatively, by 
surpassing the bound on the space-restricted random-access machine.

Our algorithms rely heavily on the random-access capabilities.
The kernel of our construction is the wavelet stack---a 
new data structure that allows us to eliminate elements while being able
to sequentially scan the active elements and skip over the eliminated
ones. This data structure only requires a constant number
of bits per element (instead of the usual $\ceils{\lg N}$ bits
required for storing indices). The wavelet stack is by no means
restricted to this particular application. We believe it would
be generally useful for prune-and-search algorithms in the
space-bounded setting. A wavelet stack comprises several layers of bit
vectors, each supporting \Rank{} and \Select{} queries in $O(1)$
worst-case time \cite{Cla96,Jac89,Mun96,RRR07}. Using the \Rank{} and
\Select{} operations, we can navigate between the layers of
the stack and perform successor queries efficiently.

\section{Previous results}

In this section we recall the main ideas of Munro-Paterson and Frederickson selection
algorithms, some of which we shall use later. We also discuss the time-space lower bound
for the selection problem in the multi-pass streaming model.
In addition, we highlight some selection algorithms for strictly limited workspace
and recent developments for integer data.
Throughout this section we assume the available workspace 
to be $\Theta(s_w)$ words, i.e.~$S = s_w \lg N$ bits. 

\subsection{Munro-Paterson selection algorithm}

Munro and Paterson \cite{MP80} outlined a selection algorithm for the multi-pass streaming model that achieves a running time of $O(N \lg s_w + N (\lg N) / \lg s_w)$, when $s_w = \Omega(\lg^2 N)$. It should be noted that this algorithm is originally designed to optimize the number of passes made over the input. 

The main idea of the algorithm is to repeatedly select two filter
elements of improving quality. 
The filters determine which elements are still to be considered. 
Any element falling outside the range of the filters is simply ignored.
After a number of iterations, there are
few enough candidates remaining in the range of the filters so that we
can find the designated element using a standard linear-time selection algorithm \cite{BFPRT73} within
the limits of the available workspace. 

In each pass a {\em sample} is constructed from the elements falling between the filters of the previous pass. 
An $s$-sample at level-$i$ is a sorted sequence of $s$ elements deterministically chosen from a population of $s \cdot 2^i$ candidates. A level-$0$ sample consists of $s = 2\ceils{s_w / (2 \lg N)}$ candidates in sorted order, and is obtained by resuming the sequential scan of the input array. A level-$i$ $s$-sample is a {\em thinning} of two $s$-samples from level-$(i-1)$, and is obtained by selecting every other element from each sample and then merging the two thinned samples. To utilize the storage efficiently, a bottom-up approach is employed to iteratively construct next-level samples once two are ready at a level. 
Thus, when an $s$-sample is constructed, there is at most one $s$-sample at each level other than the one that has just been produced. 
Let $n$ be the number of active elements at the beginning of a pass.
At the end of that pass, two improved filters are selected from the $s$-sample at level $r = \ceils{\lg(n/s)}$. 
The ranks of the two new filters with respect to the sorted sample at level~$r$ are $\lceil k/2^r \rceil -r$ and $\lceil k/2^r \rceil$. 

The total running time of the Munro-Paterson algorithm is $O(N \lg s + N (\lg N) / \lg s)$.
The actual running-time analysis of the algorithm is due to Frederickson \cite{Fre87}, whose arguments can be summarized as follows:

\begin{itemize}
\item Starting with $n$ active elements remaining in the range of two filters, the next pass will reduce the number of active elements to $O((n/s) \cdot \lg (n/s))$. 
\item The number of passes performed by the algorithm is $O((\lg N) / \lg s)$. 
\item Except for $\Theta(N)$ work done per pass to scan and compare the elements, the $O(N \lg s)$ time consumed in sorting the level-$0$ samples during the first pass dominates the rest of the work.  
\end{itemize}

\subsection{Frederickson's improved selection algorithm}
\label{subsec:fred}

As mentioned above, Frederickson \cite{Fre87} observed that the bottleneck of the Munro-Paterson algorithm is the sorting done to create the $s$-samples at level~$0$ during the first pass. Since there are $N/s$ such samples, their sorting cost accounts for the $N \lg s_w$ term in the running time.

Frederickson improved the algorithm by modifying the sampling procedure.
Using a parameter $d$, the algorithm finds the $d$-quantiles of the size-$s$ sets that are gathered at level~$0$ (instead of sorting the sets).
The execution of the algorithm can be divided into $\mathcal{P}$ {\em phases}, where $\mathcal{P}= \lg^* s -2$.  
In each phase, the algorithm performs a constant number of passes until the number of elements is reduced to $N/\lg^{(\mathcal{P})} s$. The value of $\mathcal{P}$ is decremented and a new phase is performed, repeating until $\mathcal{P}=0$.  
After each phase, the parameter is adjusted to $d = \lg^{(\mathcal{P})} s$. 
As $\mathcal{P}$ decreases by one in each phase, $d$ increases exponentially.
Each pass requires $O(N + n \lg d)$ time, where $n$ is the number of active candidates before the pass. 
Initially, only a constant number of quantiles are computed, and as the number of remaining candidates decreases the number of quantiles computed per sample increases exponentially. At the low levels, instead of thinning and merging the $d$-samples, they are simply merged such that at level~$i$ the samples have size $2^i \cdot d$. Once $2^i \cdot d \geq s$, the thinning and merging procedure is again in use; this keeps the sample size bounded by $s$. As before, the new filters are computed from a final level-$r$ sample (as in the Munro-Paterson algorithm). 
In all, the work done during each pass, and hence also during each phase, is $\Theta(N)$. 
The above procedure allows us to reduce the input size to $O(N / \lg s)$ once $d=s$, which happens after $O(\lg^* s)$ phases. 
After that, the Munro-Paterson algorithm is back to action.
Combining this outcome with that of the Munro-Paterson algorithm, the running time of the overall algorithm becomes $O(N \lg^* s_w + N (\lg N) / \lg s_w)$ for $s_w = \Omega(\lg^2 N)$.
The insights of the analysis can be summarized as follows:

\begin{itemize}
\item Starting with $n$ active elements remaining in the range of two filters, performing a phase that starts by finding the $d$-quantiles of size-$s$ sets, the number of active elements is reduced to $O((n/s) \cdot \lg (n/s) + n / d)$. 
\item After $O(\lg^* s)$ phases, there are at most $N / \lg s$ elements remaining.
\item Each phase runs in $\Theta(N)$ time.
\end{itemize}

Chan \cite{Cha10} noted that the running time can be improved to $O(N \lg^* (N / s_w) + N (\lg N) / \lg s_w)$,
by simply switching to the Munro-Paterson algorithm once $d=\min(s, N / s)$ (instead of switching when $d=s$).
The reader is encouraged to verify that this bound is better than the $O(N \lg^* s_w + N (\lg N) / \lg s_w)$ bound.

\subsection{Selection algorithms for more limited workspace}

The selection algorithms presented up to this point require at least $s_w = \Omega(\lg^2 N)$ words of workspace to be available. For $s_w = o(\lg N)$, Munro and Raman~\cite{MR96} developed an algorithm based on recursively finding the median of a block of candidates to filter the elements until the required element is found. Their algorithm runs in $O(2^{s_w} s_w!N^{1 + 1/s_w})$ time. When $s_w = 1/\varepsilon = O(1)$, i.e.~with constant number of words of workspace, the running time becomes $O(N^{1+\varepsilon})$. It is also worth noting that the total number of comparisons made by their algorithm is minimized for $s_w = O(\sqrt{(\lg N) / \lg \lg N})$, which gives a running time of $O(N^{1 + O(\sqrt{(\lg \lg N) / \lg N})})$. Munro and Raman also proved that if the elements in the input are assumed to be in random order, then their algorithm can be modified to have an average-case running time of $O(N \lg((\lg N) / \lg s_w))$.
Chan~\cite{Cha10} showed how this algorithm can be randomized so that
the assumption on the order of elements in the input is not needed.

Raman and Ramnath~\cite{RR99} improved the performance when $s_w$ is $o(\lg^2 N)$ and $\Omega(\lg N)$, by describing an algorithm that finds a pair of approximate medians and uses them to construct a three-way partition of the active elements. 
The running time of this algorithm is $O(N \lg^2 N)$ when $s_w = \Theta(\lg N)$. 
They also presented a generalization of this algorithm for smaller values of $s_w$, by describing how to determine the desired approximate median-pair with less space. The  running time of the modified algorithm is $O(s_w N^{1+1/s_w} \lg N)$, which is an improvement over Munro and Raman's algorithm~\cite{MR96} when $s_w = O(\lg N)$ and $2^{s_w}s_w! > s_w \lg N$ (e.g.~when $s_w \geq c \cdot (\lg \lg N) / \lg \lg \lg N$ for some positive constant $c$). Raman and Ramnath also described how the running time can be reduced further if more space is available. This is done by computing a set of three or more splitters instead of the pair of approximate medians, allowing the candidates to potentially be split into more than three buckets. The running time of this algorithm is $O(N \lg N + N (\lg^2 N) / \lg^2 \mathcal{I})$ when $s_w = O(\mathcal{I}^2 (\lg N) / \lg \mathcal{I})$, where $\mathcal{I} \geq 2$ is an integer parameter. This is worse than Frederickson's algorithm by a factor of $(\lg N) / \lg \lg N$ when $s_w = \Theta(\lg^2 N)$, but unlike Frederickson's algorithm, it can be applied in cases where $s_w$ is $o(\lg^2 N)$ and $\Omega(\lg N)$.

\subsection{A lower bound in the multi-pass streaming model}
Munro and Paterson \cite{MP80}, using adversarial arguments, showed that any comparison-based selection algorithm in the multi-pass streaming model must perform $\Omega((\lg N) / \lg s_w)$ passes.
Chan~\cite{Cha10}, also using adversarial arguments, proved that any deterministic comparison-based selection algorithm in this model must use either $\Omega(p)$ passes 
or $\Omega(N \lg^{(p)}(N/s_w))$ comparisons (for any $p$),
indicating that any such algorithm must take $\Omega(N \lg^*(N/s_w))$ time.
Combining the two results, it follows that any comparison-based
deterministic selection algorithm for the multi-pass streaming model
must spend $\Omega(N \lg^*(N/s_w) + N (\lg N) / \lg s_w)$ time on
either scanning or comparisons. 
Chan~\cite{Cha10} posed as an open problem whether this bound also holds for 
the space-restricted random-access model. We show that this is not the case, 
indicating that the two models are distinct in the context of deterministic selection.

\subsection{Integer selection}

Recent work by Chan, Munro and Raman~\cite{CMR-isaac} has indicated that faster selection algorithms are possible if we restrict the input elements to be a sequence of integers. When the input elements come from the universe $\{1, 2, \dots, U\}$, they presented two algorithms, one running in $O(N \lg_{s_w} U)$ using $O(s_w)$ words of space for any $s_w$ from $1$ to $N$, the other running in $O(N (\lg N) \lg_{s_w} \lg U)$ while using $O(s_w)$ words of space for any $s_w$ from $1$ to $\lg U$. 
The first algorithm determines the bits of the $k$-th element by iteratively counting the number of 1- and 0-bits among the candidates at the current bit location and comparing these counts with $k$. This approach only uses $s_w = O(1)$ words, but can be easily extended to handle $b$ bits during each pass when we have $s_w = 2^b$ extra words of workspace available. The second algorithm is significantly more involved, but the basic idea is to utilize the bit-prefixes of the integers to efficiently select approximate medians that can then be used in a well-known prune-and-search approach for the selection problem. The first algorithm is a significant improvement over existing methods for small universe sizes, whereas the second algorithm is less sensitive to the universe size and thus provides an improvement over existing algorithms for a wider range of universe sizes.

\section{Basic toolbox}

In this section we describe the basic tools used in our algorithms.

\subsection{Bit vectors with \Rank{} and \Select{} support}
A bit vector is an array of bits (0's and 1's).  Consider a bit
vector $V$ that supports the following operations:
\begin{description}
\item[$V.\Access{}(i)$:] Return the bit at index $i$, also denoted
  as $V[i]$. 
\item[$V.\Rank(i)$:] Return the number of 1-bits among the bits $V[1],
  V[2],\ldots, V[i]$.
\item[$V.\Select{}(j)$:] Return the index of the $j$-th 1-bit,
  i.e.~when the return value is $i$, $V[i]=1$ and $\Rank(i) = j$.
\end{description}
On a word RAM with $w$-bit word size, one can store a sequence of $N$
bits using $\ceils{N/w}$ words, such that any substring
of at most $w$ bits---not only a single bit---can be accessed in $O(1)$
worst-case time.

There exist several space-efficient solutions to support the above
operations in $O(1)$ worst-case time. Jacobson \cite{Jac89} showed how
to support \Rank{} and \Select{} in $O(\lg N)$ bit probes using
$o(N)$ bits in addition to the bit vector. Clark and
Munro~\cite{Cla96,Mun96} showed how to support the queries in $O(1)$
worst-case time using $O(N/\lg\lg\lg N)$ bits of extra workspace on a RAM with word 
size $\Theta(\lg N)$ bits. Raman et al.~\cite{RRR07} improved 
the space bound to $O(N (\lg \lg N)/ \lg N)$ bits, which was shown 
by Golynski~\cite{Gol07} to be
optimal provided that the bit vector is stored in plain form (using 
$N$ bits or $\ceils{N/w}$ words). The
basic idea of these solutions is to divide the input into
blocks, store the \Rank{} and \Select{} values for some specific
positions, and compute the \Rank{} and \Select{} values for the
remaining positions on the fly using: the stored values, values in some
precomputed tables, and bits in the bit vector under consideration.

Note that the requirements on the bit vectors for our use in this paper are that 
(i) the space usage must be
$O(N)$ bits, (ii) the operations must have $O(1)$ worst-case cost, 
and (iii) the construction of the supporting structures must
take $O(N)$ worst-case time. For these requirements, Chazelle \cite{Cha88}
described a simple solution to support the \Rank{} operation.  After
breaking the bit vector into words, for the first bit of each word a
\emph{landmark} is computed that is the number of 1-bits preceding
this position. Let the words be $B_1, B_2,\ldots,B_{\ceils{N/w}}$
and the landmarks be $L_1, L_2,\ldots,L_{\ceils{N/w}}$. To compute
$\Rank(i)$, we locate the corresponding word $B_j$ where $j =
\ceils{i/w}$, and calculate the offset $f$ within this word as $f=i - w
\cdot \floors{i/w}$. Then we mask the bits up to index $f$ in $B_j$
and calculate the number of 1-bits in the masked part; let this
number be $q$.  As the end result, we output $L_j + q$ as $\Rank(i)$.
The only remaining part is how to calculate the number of 1-bits in a
word, but this can be done in $O(1)$ time using precomputed tables 
of size $o(N)$ bits. (In practice, one can use the population-count
function that is a hardware primitive in most modern processors.)

Let $Q$ be the number of ones in the bit vector $V$.
To support the \Select{} operations, we construct an array of length $\floors{Q/\lg N}$,
using $O(N)$ bits, whose $j$-th entry stores the position of the $(j \ceils{\lg N})$-th one in $V$,
for $1 \le j \le \floors{Q/\lg N}$. If the difference between 
two consecutive entries in this array is at least $\lg^2 N$, then we store the 
positions of all the $\ceils{\lg N}$ ones in between the two positions using $O(\lg^2 N)$
bits. If the difference between two consecutive entries is less than $\lg^2 N$, we 
construct a complete tree with branching factor $\sqrt{\lg N}$ and 
constant height that stores the bit sequence between the two positions
at its leaves, such that \Select{} queries in this range can be 
answered in constant time. See~\cite{RRR07} for the details of such a structure.

\subsection{Wavelet stacks}
 
In a prune-and-search algorithm, where some of the answer candidates are 
repeatedly eliminated, the set of active elements can
be compactly represented using a bit vector.  The
history of the decisions made by such an algorithm can be conveniently
described using a stack of bit vectors. We call this kind of data
structure a wavelet stack because of its resemblance to a wavelet
tree~\cite{GGV03,Nav12}.
In the rest of this section,
we describe the wavelet-stack data structure in detail. In the next section, we
show how it can be used to solve the selection problem.  We believe that
this data structure will be useful in other applications as well.

Let $\sequence{x_1, x_2, \ldots, x_N}$ be a sequence of $N$ elements given in
a read-only array. Assume we want to find a specific subset of
these elements using prune-and-search elimination.  A prune-and-search
algorithm is a recursive procedure that may call itself several
times. Hence, we need a recursion stack to keep track of the
subproblems being solved. In addition to a recursion stack (with
constant-size activation records), we maintain a stack of bit vectors
to mark the active elements in the current configuration. The $i$-th bit in the bit vector 
at a given level corresponds to the $i$-th active element (in the left-to-right
order) at the level below, and a 1-bit in a bit vector at a level
indicates that the corresponding element is still active up to the current level.
Using $\Rank$ and $\Select$ operations on the bit vectors, we can  
scan the active elements at any level (and avoid scanning the 
pruned elements) faster than a left-to-right scan of the input array.

In an abstract form, the wavelet stack is a stack of bit vectors, $H$, 
that can efficiently answer two types of queries:
\begin{description}
\item[$H.\Active{}(i)$:] Return whether the element $x_i$ is
  active at the current configuration.
\item[$H.\Index{}(j)$:] Return the index of the $j$-th active element,
  i.e.~the index of the element corresponding to the $j$-th 1-bit of
  the top-most bit vector.
\end{description}

To fully understand these operations, we have to consider a concrete
implementation of a wavelet stack (see Fig.~\ref{fig:stack}). A wavelet stack is a
hierarchy of bit vectors. The bottom-most level stores one bit per
element, since at the beginning all elements are potential answers (i.e.~active). 
If we have $y$ 1-bits at level~$\ell$, the bit vector at level $\ell + 1$ is
of size $y$. Therefore, the bit vectors become smaller and smaller as
we eliminate more elements from further consideration.

\begin{figure}[tb!]
\begin{center}
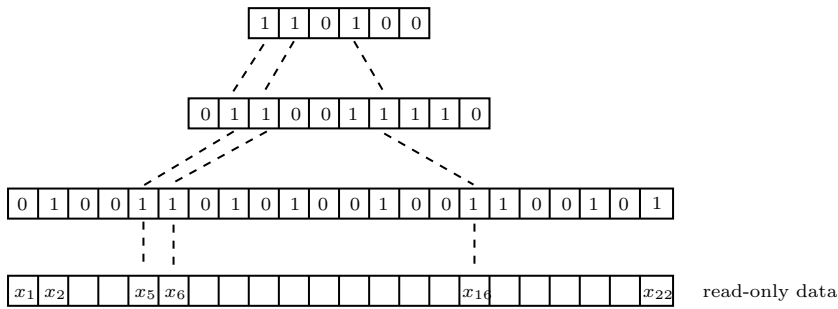
\end{center}
\caption{A wavelet stack for an array of 22 elements. Only elements
  $x_5$, $x_6$, and $x_{16}$ are active at the top-most level.\label{fig:stack}}
\end{figure}

The two operations have a nice symmetry: the operation \Active{} can be
supported by traversing up from the bottom to the top of the stack, and 
the operation \Index{} can be supported by traversing down from
the top to the bottom of the stack. To implement \Active{}($i$), we
compute $\Rank(i)$ at the bottom-most level, which gives us
the index to access in the bit vector immediately above. Continuing upwards and
relying on \Rank{}, we either reach a level where the bit corresponding to the index
value is $0$, indicating that the element $x_i$ is not active any more, or 
reach the top-most level where the bit value is $1$, indicating that $x_i$ is still active. 
To implement \Index{}($j$), i.e.~to return the $j$-th active element at the top-most
level, we start from the top-most level and compute \Select{}$(j)$. 
Then we use the returned index at the level immediately below. This way,
we can proceed down using \Select{} until we reach the bottom-most level. 
The index returned at this level is the index of the $j$-th active element in the input array.

We can summarize the performance of the data structure as follows:

\begin{theorem}
\label{wavelet-stack}
Assume that we have built a wavelet stack of height $h$ for a
read-only array of $N$ elements. Furthermore, assume that at each
level we have succeeded in eliminating a constant fraction of the
elements.
\begin{enumerate}
\item The data structure requires $\Theta(N)$ bits in total.
\item The total time used in the construction of the data structure is
  $\Theta(N)$. 
\item Both \Active{} and \Index{} operations
  take $O(h)$ worst-case time.
\end{enumerate}
\end{theorem}

\begin{proof}
Since the number of bits needed at each level is only a constant
fraction of that needed at the level below, for some constant $c < 1$,
the total number of bits of all the bit vectors is bounded by
$\sum_{i=1}^{h} c^{i-1} N = \Theta(N)$ bits. The supporting structures 
for \Rank{} and \Select{} also sum to $O(N)$ bits. 
Since the length of the bit vectors is not known beforehand and
since their sizes may vary, we can allocate a header storing
references to a single bit vector that contains the bits stored at all
levels together. This header will only require $O(h \lg N) = O(\lg ^2 N)$ bits.

The construction of a bit vector, including the supporting structures, can be done in time linear in the vector size. 
The construction time of the wavelet stack can also be expressed as a geometric series,
and is thus $\Theta(N)$. Since the structure has
$h$ levels, and the \Rank{} and \Select{} operations take $O(1)$
worst-case time at each level, it can support \Active{} and \Index{} operations in $O(h)$ time.
\qed
\end{proof}

\section{Selection with $\Theta(N)$ bits}

In this section we show how to utilize
the prune-and-search algorithm of Blum et al.~\cite{BFPRT73}
(also described in~\cite[Section 9.3]{CLRS09}) such that it only
requires $\Theta(N)$ bits of space---instead of $\Theta(N)$ words---but
still runs in $\Theta(N)$ time.

The main idea of the algorithm is to select an element from the
set of active elements, and use it to make the set of candidates
smaller (by a constant factor). This is done repeatedly until the 
required element is found. In the variant considered here we use a 
wavelet stack to keep track of the decisions made by the algorithm. 
The $k$-th smallest among $n$ active elements is found as follows.
\begin{enumerate}
\item A new bit vector $V$ is pushed onto the top of the
  wavelet stack. The size of this bit vector equals the number of
  the currently active elements $n$.
\item \label{groups}Divide the sequence of $n$ elements into groups of size $\ceils{N/\lg N}$ in the same order 
  as in the input array (the last group may have less elements). 
  Find the median of each of the, at most, $\ceils{\lg N}$ groups, 
  by holding the indices of the elements of each group in the available workspace and 
  applying a standard linear-time selection algorithm \cite{BFPRT73}.
\item\label{median-of-medians} Store the indices of the found medians in the available workspace. Find the median $x_{\mu}$ of
  the medians of the groups, applying a standard linear-time selection algorithm \cite{BFPRT73}. 
\item\label{filter-case} Scan through the active elements and
  count the number $\sigma$ of those smaller than $x_{\mu}$.
  If $k = \sigma + 1$, stop and return $\mu$ as an answer. 
  If $k \leq \sigma$, mark the elements smaller than $x_{\mu}$ as the only active elements in $V$ 
  and recursively compute the $k$-th smallest of
  these elements.  Otherwise, if $k > \sigma + 1$, mark the elements
  larger than $x_{\mu}$ as the only active elements in $V$ and set $k$ to $k - \sigma - 1$
  before the recursive call.
\item When $n$ is at most $\ceils{N/\lg N}$, copy the indices of the active 
  elements into the available workspace (releasing the space used by the 
  wavelet stack), and find the $k$-th smallest element
  using a standard linear-time selection algorithm \cite{BFPRT73}.
\end{enumerate}

The analysis of this algorithm is
almost identical to that of the original algorithm of 
Blum et al.~\cite{BFPRT73}. The key point is that, even
though the input is in a read-only array, we do not waste time in
browsing through the elements that have already been eliminated, as we rely
on the \Rank{}-\Select{} operations supported by the bit vectors to 
scan through the active elements efficiently. 
The only overhead is that when we want to access an element we have to
traverse down the wavelet stack.

The performance of the algorithm is summarized in the following
theorem:

\begin{theorem}
\label{thm:linearbits}
The $k$-th smallest of $N$ elements in a read-only array can be found in
$\Theta(N)$ time using $\Theta(N)$ extra bits in the worst case.
\end{theorem}

\begin{proof}
At step \ref{groups} of the algorithm, the number of elements of each group is at most $\ceils{N/\lg N}$.
In accordance, the indices of all the elements of a group can be simultaneously stored
using $O(N)$ bits of workspace. A standard linear-time selection algorithm can then be applied on each group
at a time. Similarly, the median of medians can be found in linear time at step \ref{median-of-medians}
of the algorithm within the storage limitations of the available workspace.  

The main observation is that the number of elements pruned at step \ref{filter-case} of the algorithm 
is at least $n/4 - O(\lg N)$, where $n$ is the number of active elements before the pruning.
Hence, the number of the surviving active elements before the next recursive call is at most $3n/4 + O(\lg N)$.
Since $n$ must be larger than $\ceils{N/\lg N}$ to perform a recursive call, thus the 
number of active elements before the $i$-th recursive call is at most $(3/4 + \varepsilon)^{i-1} N$, where $\varepsilon$ is $o(1)$ compared to $N$.
Following Theorem \ref{wavelet-stack}, the total size of all the bit vectors and the accompanying structures of the wavelet stack is $\Theta(N)$ bits and its construction time is $\Theta(N)$. To calculate the time for scanning over the active elements, we note that getting the successor of each active element at the $i$-th recursive call consumes $\Theta(i)$ time. It follows that the total time for scanning over the active elements in all the recursive calls is $\Theta(\sum_{i \geq 1} i \cdot (3/4+\varepsilon)^{i-1} \cdot N) = \Theta(N)$.   
\qed
\end{proof}

\section{Selection with $\Theta(S)$ bits}

In this section we extend our algorithm to handle the more general case of 
using a workspace of $\Theta(S)$ bits, where $\lg^3{N} \leq S \leq N$.
The main idea is to use Frederickson's algorithm \cite{Fre87} to prune the elements, and stop its execution when the number of active elements is at most $S$.  To complete the selection process, we resume pruning using an $O(N)$-time algorithm that we present next.

We use the following lemma, which is based on Frederickson's 
algorithm discussed in Section~\ref{subsec:fred}.
We refer the reader to \cite{Fre87} for the full details.

\begin{lemma}
\label{Fred}
The number of active elements can be reduced from $N$ to $S$ during the
initial phases of Frederickson's algorithm 
in $O(N \lg^* (N / S))$ worst-case time, assuming $S = \Omega(\lg^3 N)$.
\end{lemma}

We now generalize our selection algorithm from the previous section
to obtain time-space trade-offs. In particular, we describe a selection
algorithm that takes $O(N \lg^* (N / S) + N (\lg N) / \lg S)$ time given
only $\Theta(S)$ bits of workspace, where $\lg^3 N \le S \le N$.

Recall that Frederickson's selection algorithm takes 
$O(N \lg^*(N (\lg N)/S) + N (\lg N) / \lg S)$ time, 
for any $S = \Omega(\lg^3 N)$.
If $S \leq \sqrt{N \lg N}$, we simply use Frederickson's algorithm all
the way, and the resulting running time is as claimed $O(N \lg^*{(N (\lg N)/S}) + N
(\lg N) / \lg S) = O(N \lg^* (N / S) + N (\lg N) / \lg S)$.
From now on we assume that $S > \sqrt{N \lg N}$. We apply a trimmed
execution of Frederickson's algorithm as specified in Lemma~\ref{Fred}.  
The outcome is two filters that guard the, at most, $S$
active elements. Consequently, we are left with the task of selecting the
designated element among those candidates.

Using a wavelet stack and a bit vector supporting
\Rank{} and \Select{} queries, we can finish the pruning in $O(N)$ time.  
We divide the input sequence (consisting of $N$ elements) into $S$ buckets, 
where the $u$-th bucket consists of the elements from the input sequence with indices in the
range $[(u-1) \cdot \ceils{N / S} + 1 \twodots u \cdot \ceils{N  / S}]$, 
for $1 \le u \le S$ (except possibly the last
bucket).  We create the \emph{count vector} $C$---a
static bit vector that indicates the number of active elements originally
contained in each bucket after the execution of Frederickson's
algorithm. The count vector $C$ should also support \Rank{} and
\Select{} queries efficiently.  We store these counts encoded in unary, 
using a 0-bit to mark the border between every two consecutive buckets.
Since a total of at most $S$ candidates need to be stored, the count
vector $C$ contains at most $S$ ones. Since we have exactly $S$
buckets, $C$ contains $S - 1$ zeros. The count vector thus uses $\Theta(S)$ bits.
In addition, we create and maintain a wavelet stack $H$---an element hierarchy where
each bit corresponds to an element among those whose values fall
in the range of the filters. Since there are at most $S$ such elements, the
wavelet stack $H$ uses $O(S)$ bits as well.  While our algorithm is in action,
the wavelet stack is to be updated to indicate the elements that are
currently surviving the pruning phases. 

We can now iterate efficiently through the active elements.
Let $i-1$ be the rank of the element that has just been considered in our iterative scan within the currently active elements.
First, we find the index $j$ of the next element to be considered within the wavelet stack. 
For that we compute 
\[j= H.index(i),\] 
which is the index of the element we are looking for with respect to those falling
between the two filters inherited from Fredrickson's algorithm. 

The position $a$ of this element in the count vector $C$ is
\[a = C.select(j).\]
The difference between the position of a bit in the count vector, $C$, and the \Rank{} of that bit plus one is the index of the bucket that contains the corresponding element. We compute the index $u$ of the bucket containing this element as
\[u = a - j + 1.\]
If $u>1$, we calculate the index $z$ that corresponds to the 0-bit
resembling the border between the $(u-1)$-th and $u$-th buckets in $C$. This is
done by utilizing, $\bar{C}$, the complement vector of $C$ to get
\[z = \bar{C}.select(u-1).\] 

We finally determine the position $g$ of the sought element among
Frederickson's candidates within the $u$-th bucket as

$$g = \left\{
\begin{array}{l}
a \mbox{~~~~~~~if $u=1$} \\
a - z \mbox{~~if $u>1$.}
\end{array}
\right.$$

If the preceding alive element in the scan was also from bucket $u$, we continue scanning the elements of the $u$-th bucket from where we stopped.
Otherwise, we jump to the beginning of the $u$-th bucket, i.e.~to the element whose index is $(u-1) \cdot \ceils{N / S} + 1$ in the input array.
We sequentially scan the elements of this bucket, discard the ones falling outside the filters and count the others, until locating the $g$-th element among them; this is the one we are looking for.

We can now proceed as in the $\Theta(N)$-bit solution. Starting with the elements surviving Frederickson's algorithm, we recursively determine the median-of-medians and use it to perform the pruning. During this process, we keep the wavelet stack up to date as before.
The pruning process continues until only one bucket containing active elements remains, at such point only $O(N / S)$ elements are active. 
Since this branch of the algorithm is employed only when $S = \Omega(\sqrt{N \lg N})$, the indices of the active elements can fit in the allowable workspace, each in $O(\lg N)$ bits, and we continue the selection in linear time.

Since we are operating on buckets, we might have to spend $\Theta(N/S)$ time for scanning per bucket. 
However, we note that initially there is at most $S$ candidates and accordingly at most $S$ buckets. 
Since we prune a constant fraction of the candidates in each iteration, we also reduce the bound on the number of the remaining buckets (those having at least one active element each) by the same constant fraction.
Because we skip the buckets that have no active elements, the work done per pass to iterate over the buckets that have at least one active element can be bounded, as elaborated in the next lemma.

\begin{lemma}
Given a read-only input array with $N$ elements, and two filters, such that at most $S$ elements lie in the range of the filters. 
If $S = \Omega(\sqrt{N \lg N})$, we can solve the selection problem in $O(N)$ time.
\end{lemma}

\begin{proof}
In each pruning iteration we spend time proportional to the number of buckets remaining, while scanning the elements in these buckets and comparing them with the filters. The number of active elements before we apply the $i$-th pruning iteration of the median-of-medians algorithm is at most $S/c^{i-1}$, for some constant $c > 1$. Obviously, the number of buckets that have active elements cannot exceed the number of elements. It follows that, throughout all the passes of the algorithm, the number of scanned buckets is at most $O(\sum_{i\geq 1} S/c^{i-1}) = O(S)$. Accordingly, the overall work done in scanning these buckets is $O(N)$. Once we have $O(N/S)$ elements remaining, as $S = \Omega(\sqrt{N \lg N})$, we can continue the selection process in the available workspace in $O(N/S)$ time.
\qed
\end{proof}

The main result of this paper is summarized in the upcoming theorem.
 
\begin{theorem}
\label{thm:ram}
Given a read-only array of $N$ elements and a workspace of $\Theta(S)$ bits
where $\lg^3 N \leq S \leq N$, it is possible to solve the
selection problem in $O(N \lg^{*} (N /S) + N (\lg N) / \lg S)$
worst-case time in the space-restricted random-access model.
\end{theorem}

Theorem~\ref{thm:ram} implies that, in the read-only space-restricted
setting, Chan's lower bound \cite{Cha10} for selection in
the multi-pass streaming model does not apply to the
random-access model.

\section{Conclusions}

We showed that, given an array of $N$ elements in a read-only memory,
the $k$-th smallest element can be found in $\Theta(N)$ worst-case time
using $\Theta(N)$ bits of extra space.  We also generalized our
algorithm to run in $O(N \lg^{*}(N/S) + N (\lg N) / \lg S)$ time using
workspace of $\Theta(S)$ bits, $\lg^3 N \leq S \leq N$.  Our main purpose
was to show that the lower bound proved by Chan \cite{Cha10} for the
multi-pass streaming model can be surpassed in the space-restricted
random-access model.

In the read-only setting, the selection problem has been studied since
1980 \cite{MP80}. In contrast to sorting, the exact complexity of
selection is still open. The time-space trade-off for sorting is known
to be $\Theta(N^2/S + N\lg S)$ \cite{Bea91,PR98}, where $S$ is the
size of the workspace in bits, $\lg N \leq S
\leq N/\lg N$. The optimal bound for sorting can even be realized
using a natural priority-queue-based algorithm~\cite{AEK13}.

Subsequent to our work, Chan et al.~\cite{CMR-isaac} considered
the selection problem in the space-restricted random-access model
when the elements of the input array are integers, and gave ``improved" bounds for this 
case. Chan et al.~\cite{CMR-soda} also considered the selection (and 
also sorting) problem in another model, called the {\em restore model},
where the input array is allowed to be modified during the process of 
answering a query, but after the query is answered it has to be restored 
to its original state. They used the result of Theorem~\ref{thm:linearbits}, 
and obtained a linear-time selection algorithm with logarithmic amount of 
extra space. The selection problem in the restore model has also been 
previously considered by Katajainen and Pasanen~\cite{KP94}, who gave 
a linear-time algorithm that uses a linear number of extra bits
for the case when the elements are indivisible (i.e.~they can only be compared). 
Settling the exact complexity of the selection problem in different computational 
models is still an interesting, partially open problem.

\section*{References}

\end{document}